\documentclass[12pt]{article}

\usepackage[hmargin=2cm,vmargin=2cm]{geometry}

\usepackage{amssymb}
\usepackage{amsmath,amsthm}
\usepackage{tikz}
\usepackage[margin=1cm]{caption}
\usepackage{verbatim}

\newcommand{\La}{\mathcal{L}}
\newcommand{\Z}{\mathcal{Z}}

\newcommand{\R}{\mathbb{R}}
\newcommand{\C}{\mathbb{C}}
\newcommand{\N}{\mathbb{N}}
\newcommand{\ud}{\mathrm{d}}
\newcommand{\UI}{\mathrm{I}}
\renewcommand\Re{\operatorname{Re}}
\newcommand{\G}{\mathit{\Gamma}}

\newtheorem{thm}{Theorem}
\newtheorem{lemma}{Lemma}
\newtheorem{cor}{Corollary}

\makeatletter
\renewcommand\@biblabel[1]{#1.}
\makeatother

\title{Relating Zeta Functions of Discrete and Quantum Graphs}
\author{Jonathan Harrison and Tracy Weyand}
\date{}

\newcommand{\Address}{{
  \bigskip
  \footnotesize

 \noindent J.~M.~Harrison: Department of Mathematics, Baylor University, One Bear Place \# 97328, Waco, TX 76798-7328; \texttt{jon\_harrison@baylor.edu}\\
 
 \noindent T.~Weyand: Department of Mathematics, Baylor University, One Bear Place \# 97328, Waco, TX 76798-7328; \texttt{weyand@rose-hulman.edu}\\ Currently At: Department of Mathematics, Rose-Hulman Institute of Technology, 5500 Wabash Avenue, Terre Haute, IN 47803 

}}

\begin{document}

\maketitle
\abstract{
We write the spectral zeta function of the Laplace operator on an equilateral metric graph in terms of the spectral zeta function of the normalized Laplace operator on the corresponding discrete graph. To do this, we apply a relation between the spectrum of the Laplacian on a discrete graph and that of the Laplacian on an equilateral metric graph. As a by-product, we determine how the multiplicity of eigenvalues of the quantum graph, that are also in the spectrum of the graph with Dirichlet conditions at the vertices, depends on the graph geometry.  Finally we apply the result to calculate the vacuum energy and spectral determinant of a complete bipartite graph and compare our results with those for a star graph, a graph in which all vertices are connected to a central vertex by a single edge.\\

}

\section{Introduction}

Zeta functions are widely studied in graph theory where, for example, the Ihara zeta function associated to a finite graph is defined by an Euler product over all backtrack-less primitive closed loops, see e.g. \cite{Has89, StaTer96, Sun86}.
In mathematical physics, on the other hand, quantum graphs provide an important model to investigate phenomena associated with complex quantum systems.  In a quantum graph, edges of the graph correspond to intervals with a differential operator, typically the Laplace or Schr\"{o}dinger operator, acting on functions on the intervals.  Quantum graphs are employed in diverse areas including Anderson localization, quantum chaos, nanotechnology, and the theory of photonic crystals; see \cite{BKbook} for an introduction. In many of these applications, it is the spectral properties of the graphs that are of particular interest.  A spectral zeta function is
\begin{equation*}
\sum_j\,\!^{^\prime} \lambda_j^{-s}
\end{equation*}
where $\{\lambda_j\}$ is the point spectrum of a self-adjoint operator and the prime indicates that zero eigenvalues are excluded.  For finite quantum graphs, such a zeta function can be written as a sum over periodic orbits using the trace formula.  Alternatively, it can be formulated in terms of the vertex conditions using a contour integral approach for the Laplace \cite{HarKir11}, Schr\"odinger \cite{HarKirTex12}, and Dirac \cite{HarWeyKir16} operators.  Spectral zeta functions of lattice and torus graphs have also been studied in the case of discrete graphs, where they are seen to inherit properties associated with the Riemann zeta function \cite{FriKar16}.

While the spectrum of the Laplacians of discrete graphs and quantum graphs appear quite different, there is a relation between them \cite{Below85, Kuc04, Pank06, BKbook}. In this paper, we use this to relate the spectral zeta function of the Laplace operator on an equilateral quantum graph (a graph where every edge has the same length) to the spectral zeta function of the normalized Laplace operator on a discrete graph.  This should be seen as a first step in connecting the literature on the zeta functions of discrete graphs with spectral properties of quantum graphs.

The article is organized as follows. In Section \ref{sec: background}, we define the discrete and quantum spectral zeta functions of the respective Laplace operators. We state and prove our main result relating the quantum spectral zeta function to the corresponding discrete spectral zeta function in Section \ref{sec: main}. In Section \ref{sec: applications}, we apply the result to compute the vacuum energy and spectral determinant of a complete bipartite graph. The results are compared to those for a star graph, a graph in which all vertices are connected to a central vertex by a single edge.

\section{Background}\label{sec: background}

A \emph{discrete graph} $G$ consists of a set of vertices $\mathcal{V}$ and a set of edges $\mathcal{E}$ that connect pairs of vertices, so an edge $e=(u,v)$ for $u,v\in \mathcal{V}$; see e.g. Figure \ref{fig: star}. In this paper, we consider finite discrete connected graphs that have a finite number of vertices and edges. We denote the number of vertices by $V = |\mathcal{V}|$ and the number of edges by $E = |\mathcal{E}|$. The \emph{first Betti number} of the graph is $\beta : = E - V + 1$, the number of independent cycles on $G$. The \emph{degree} of a vertex $v$, denoted $d_v$, is the number of edges connected to $v$.  A \emph{bipartite graph} is a graph where the vertex set can be split into two disjoint parts, $\mathcal{V}=\mathcal{U}\cup \mathcal{W}$ with $\mathcal{U}\cap \mathcal{W} = \emptyset$, such that every edge $e=(u,w)$ with $u\in \mathcal{U}$ and $w\in \mathcal{W}$.

Functions on a discrete graph $G$ take values at the vertices, and hence, are represented by vectors in $\C^{V}$. Operators that act on these functions can be represented as $V \times V$ matrices. For example, the \emph{normalized} (harmonic) \emph{Laplace operator} on a discrete graph is defined as
\begin{equation}
(\Delta f)(v) = f(v) - \frac{1}{d_v}\displaystyle\sum_{u\sim v}f(u),
\end{equation}
or alternatively, $\Delta$ is the $V \times V$ matrix whose entries are
\begin{equation}
\Delta_{u,v} = \left\{\begin{array}{ll}
1 & \mbox{if } u = v\\
-\frac{1}{d_v} & \mbox{if } u \sim v\\
0 & \mbox{otherwise}.
\end{array}\right.
\end{equation}
We denote the eigenvalues of $\Delta$ by $\lambda_1 \leq \lambda_2 \leq \cdots \leq \lambda_{V}$.

A \emph{metric graph} $\G$ is a discrete graph on which each edge $e\in \mathcal{E}$ is assigned a length $L_e$ and associated with the interval $[0,L_e]$; here we assume that every length $L_e$ is finite. The orientation of the coordinate $x_e$ is arbitrary; our results are independent of this choice. However, for clarity, given an edge $e =(u,v)$ connecting vertices $u$ and $v$ we fix the order of $u$ and $v$ and assign the orientation $x_e =0$ at $u$ and $x_e = L_e$ at $v$.  A function $f$ on a metric graph is defined by a collection of functions $\{f_e\}_{e\in \mathcal{E}}$, one on each interval. In this paper, we consider  \emph{equilateral metric graphs} where every edge has the same length $L$. Given a discrete graph $G$, we make a corresponding equilateral metric graph $\G$ by assigning length $L$ to each edge.

A \emph{quantum graph} is a metric graph equipped with a self-adjoint differential operator. Here we consider the standard \emph{Laplace operator}, which is defined as
 \begin{equation}
 \La f_e = -\frac{\ud^2f_e}{\ud {x_e}^2},
\end{equation}
together with the \emph{Neumann-Kirchhoff vertex conditions}
\begin{equation}\label{eq: NK vertex conditions}
\left\{ \begin{array}{l}
f \mbox{ is continuous at all vertices } v \in \mathcal{V} \mbox{ and}\\
\displaystyle \sum_{e\in \mathcal{E}_v} f_e'(v) = 0 \mbox{ at all vertices } v \in \mathcal{V}
\end{array}\right.
\end{equation}
where $\mathcal{E}_v$ is the set of all edges attached to vertex $v$.  When evaluating the derivative at a vertex, by convention we consider the derivative to be taken into the edge $e$ (away from the vertex). The second Sobolev space on an interval $[a,b]$ is the set of all functions such that the function, its first derivative, and its (weak) second derivative are all in $L^2([a,b])$. The second Sobolev space on $\G$ is then the direct sum of second Sobolev spaces on the intervals,
\begin{equation}
H^2(\G) = \bigoplus_{e \in \mathcal{E}} H^2([0,L]),
\end{equation}
and the domain of $\La$ is all functions in $H^2(\G)$ that satisfy the vertex conditions \eqref{eq: NK vertex conditions}. The Laplacian with these vertex conditions has real non-negative eigenvalues \cite{BKbook}, and hence, they can be written as $0 \leq k_1^2 \leq k_2^2 \leq \ldots$ where $k_j \in \R$. We also consider the self-adjoint operator $-\frac{\ud^2}{\ud {x_e}^2}$ with the \emph{Dirichlet vertex conditions} $f_e(v) = 0$ for all vertices $v \in \mathcal{V}$. The set of all eigenvalues of this operator is called the \emph{Dirichlet spectrum}.

The \emph{spectral zeta function} is a generalization of the Riemann zeta function where nonzero eigenvalues of an operator take on the   role of the integers. Let $Z(s)$ denote the spectral zeta function of the normalized Laplacian $\Delta$ on a discrete graph $G$ and $\Z(s)$ the spectral zeta function of the Laplacian $\La$ on the corresponding equilateral quantum graph $\G$. Then
\begin{equation}\label{eq: spectral zeta functions}
Z(s) = \sum_{j=1}^{V}\,\!^{^\prime} \lambda_j^{-s} \hspace{.5cm}\mbox{and}\hspace{.5cm} \Z(s) = \sum_{j=1}^\infty\,\!^{^\prime} k_j^{-2s}
\end{equation}
where the prime denotes that the sum is taken over nonzero eigenvalues. As written in \eqref{eq: spectral zeta functions}, the domain of $\Z(s)$ is $\Re(s) > \frac{1}{2}$; a domain that will be extended subsequently.

\section{Relation between the discrete and quantum spectral zeta functions}\label{sec: main}

In this paper, we prove the following relation between the spectral zeta functions of discrete and quantum graphs.

\begin{thm}\label{thm: main}
Suppose that $G$ is a discrete graph and let $\G$ be the corresponding equilateral metric graph where each edge has length $L$. Then, for $\Re(s) < 0$, the quantum spectral zeta function of $\La$ is
\begin{align}\label{eq:main}
\Z(s) &= \displaystyle \frac{2L^{2s}\Gamma(1-2s)}{\pi}\sin(s\pi)\sum_{n=1}^\infty\sum_{r=0}^n (-2)^r\frac{n^{2s}}{n+r}\binom{n+r}{2r}Z(-r)\nonumber\\
&\hspace{.5cm} + \left(4^s(\beta - 1) + 2\right) \left(\frac{L}{2\pi}\right)^{2s}\zeta_R(2s) \ .
\end{align}
$Z(s)$ is the discrete spectral zeta function of $\Delta$, $\zeta_R(z)$ is the Riemann zeta function, and $\beta = E - V + 1$ is the first Betti number.
\end{thm}

\noindent Note that, in equation \eqref{eq:main}, $\Gamma(\cdot)$ is the gamma function and not the metric graph.

To obtain this relation, we rely on the following theorem from \cite{Below85} (also see \cite{Kuc04,Pank06,BKbook}) which relates the eigenvalues of the normalized Laplace operator acting on a discrete graph to the eigenvalues of the Laplace operator acting on the corresponding equilateral metric graph.
\begin{thm}\label{thm: equivalence}
Suppose that $G$ is a discrete graph and let $\G$ be the corresponding equilateral metric graph where each edge has length $L$. If $ k^2$ is not in the Dirichlet spectrum of the Laplace operator $\La$ acting on $\G$, then
\begin{equation}
k^2 \in \sigma(\La) \iff  1-\cos(kL) \in \sigma(\Delta)
\end{equation}
where $\Delta$ is the normalized Laplace operator acting on $G$ and $\sigma(\cdot)$ denotes the spectrum of the operator.
\end{thm}

\subsection{Dirichlet eigenvalues}

Enforcing the Dirichlet vertex conditions on a quantum graph is equivalent to breaking the graph at each vertex, forming $E$ disconnected intervals. Therefore, the Dirichlet eigenvalues of the quantum graph are the eigenvalues of the differential equation
\begin{equation}
-f''(x) = \lambda f(x), \hspace{.5cm} f(0) = f(L) = 0,
\end{equation}
which are $\left(\frac{n\pi}{L}\right)^2$, $n \in \N$.

We need to determine which of the Dirichlet eigenvalues are also eigenvalues of $\La$.

\begin{lemma}\label{lemma: Dirichlet}
The multiplicity of the eigenvalue $\left(\frac{n\pi}{L}\right)^2$, $n \in \N$, in the spectrum of $\La$ for an equilateral quantum graph is,
\begin{itemize}
\item[(i)] $(\beta-1)+2\dim \ker (\Delta)$ when $n$ is even, and
\item[(ii)] $(\beta-1)+2\dim \ker (\Delta-2\UI)$ when $n$ is odd,
\end{itemize}
where $\Delta$ is the normalized Laplace operator acting on the corresponding discrete graph $G$.
\end{lemma}
While we state the lemma in the form that will be most useful subsequently, $2$ is an eigenvalue of $\Delta$ if and only if the graph is bipartite \cite{Chungbook}.  Consequently, for a connected graph, a straightforward corollary is,
 \begin{cor}\label{cor: Dirichlet}
The multiplicity of the eigenvalue $\left(\frac{n\pi}{L}\right)^2$, $n \in \N$, in the spectrum of $\La$ for a connected equilateral quantum graph is
\begin{itemize}
\item[(i)] $\beta+1$ when $n$ is even.
\item[(ii)] either $\beta+1$ if the graph is bipartite or $\beta-1$ otherwise, when $n$ is odd.
\end{itemize}
\end{cor}
Note that, if $\beta=0$ the graph is a tree, which is bipartite.
In the case of quantum graphs where the edge lengths are incommensurate, the spectrum is often studied via a secular equation whose roots are the square roots of the eigenvalues up to multiplicity.  For an equilateral graph the multiplicity of $n\pi/L$ as a root of the secular equation for even $n$ is the same as the multiplicity of zero as a root of the secular equation.  In \cite{FulKucWil07} Fulling, Kuchment and Wilson obtain this algebraic multiplicity of zero as a root of the secular equation for graphs with general vertex conditions, Corollary 23.  In particular, with Neumann-Kirchoff vertex conditions they show this multiplicity is $2-V+E=\beta +1$ in agreement with the corollary.  We will present a direct proof of the whole lemma. 
\begin{proof}
For an eigenvalue $\left(\frac{n\pi}{L}\right)^2$ of $\La$,
the eigenfunction on each edge $e$ has the form
\begin{equation}
f_e(x_e) = a_e \cos \left(\frac{n\pi}{L} x_e \right) +b_e \sin \left(\frac{n\pi}{L} x_e \right) \ .
\end{equation}

\emph{Part (i): n even:} We first deal with the case when $n$ is even.  Then $f_e(0)=f_e(L)=a_e$ for all $e$, and the continuity of $f$ on the connected graph requires $a_e=a$ for some constant $a$ on a connected component. The number of connected components is $\dim \ker (\Delta)$.

Furthermore, $f'_e(0)=f'_e(L)=\left(\frac{n\pi}{L}\right) b_e$, and hence, the conditions
\begin{equation}\label{eq:derivative}
\sum_{e\in \mathcal{E}_v} f_e'(v) = 0
\end{equation}
at each vertex $v \in \mathcal{V}$ are $V$ linear conditions on the vector $\mathbf{b}=(b_1,\dots,b_E)$ of coefficients of the sine functions. 
We can write these conditions in matrix form $Q \mathbf{b} = \mathbf{0}$ where $Q_{ve}=1$ if $x_e=0$ at $v$, $Q_{ve}=-1$ if $x_e=L_e$ at $v$, and $Q_{ve}=0$ otherwise.
To determine the dimension of $\ker (Q)$, note that,
\begin{equation}
\dim \ker (Q) - \dim \ker (Q^T) = E - V =\beta -1
\end{equation}
where $\beta$ is the first Betti number of the graph, the number of independent cycles.
If we define a diagonal matrix of the vertex degrees $D=\textrm{diag} \{d_1,\dots,d_V \}$ then,
\begin{equation}
\Delta=D^{-1} (QQ^T) \ .
\end{equation} 
The kernel of $\Delta$ is the kernel of $Q^T$ and hence,
\begin{equation}
\dim \ker (Q) = (\beta-1) + \dim \ker (\Delta) \ .
\end{equation}

\emph{Part (ii): n odd:} In the case where $n$ is odd, $f_e(0)=a_e$ and $f_e(L)=-a_e$.  Hence the continuity of $f$ on the connected graph requires that we choose $a_e=a$ or $a_e=-a$ for some constant $a$ so that the sign of the solution alternates at adjacent vertices.  This is possible if and only if the graph contains no cycles with an odd number of edges.  A graph where there are no cycles with an odd number of edges is bipartite \cite{Bipartitebook}.  Also, $2$ is an eigenvalue of $\Delta$ if and only if $G$ is bipartite \cite{Chungbook}.  Hence, the dimension of the subspace of solutions spanned by the cosine functions is the dimension of 
$\ker (\Delta -2\UI)$, the number of connected bipartite components of the graph.

Similarly, $f'_e(0)=-f'_e(L)=\left(\frac{n\pi}{L}\right) b_e$.
As in the even $n$ case, the vertex conditions (\ref{eq:derivative}) are $V$ linear conditions on the vector $\mathbf{b}=(b_1,\dots,b_E)$ of coefficients of the sine functions.  Writing the conditions in matrix form,
$M \mathbf{b} = \mathbf{0} $
where $M_{ve}=1$ if the edge $e$ is connected to $v$ and $M_{ve}=0$ otherwise.  
\begin{equation}
\dim \ker (M) - \dim \ker (M^T) = E - V =\beta -1
\end{equation}
In this case,
\begin{equation}
\Delta=2\UI - D^{-1} (MM^T) 
\end{equation} 
and $\ker (M^T) = \ker (\Delta -2\UI)$.  Hence 
\begin{equation}
\dim \ker (M) = (\beta -1) + \dim \ker (\Delta-2\UI) \ .
\end{equation}
Then combining the eigenfunctions spanned by the cosine and sine functions respectively provides the result.
\end{proof}

\subsection{Proof of Theorem \ref{thm: main}}

 First we will analyze $Z(s)$, the discrete spectral zeta function. Since all eigenvalues of $\Delta$ lie in the closed interval $[0,2]$ \cite{Chungbook}, we know that every eigenvalue can be written as $1-\cos(kL)$ for some $k \in \left[0,\frac{\pi}{L}\right]$. We define $K = \{k_j\}_{j=1}^V$ to be the set $0 = k_1 < k_2 \leq \ldots \leq k_V \leq \frac{\pi}{L}$ such that $1-\cos(k_jL) \in \sigma(\Delta)$ for each $k_j$. Each $k_j$ need not be distinct; in fact, if $1-\cos(kL) \in \sigma(\Delta)$ has multiplicity $n$, we require that $k$ appear in $K$ $n$ times. We know that $k_1 = 0$ and $k_1 < k_2$ because zero is an eigenvalue of $\Delta$ of multiplicity one for every connected discrete graph $G$.  We can write the discrete spectral zeta function as
\begin{equation}\label{eq: discrete spectral zeta}
Z(s) = \sum_{j=2}^{V} (1-\cos(k_jL))^{-s}
\end{equation}
where the sum begins at $j=2$ to avoid the zero eigenvalue of $\Delta$.

Now we will relate the eigenvalues of $\Delta$ to the eigenvalues of $\La$. If $k \neq 0,\frac{\pi}{L}$ and $1 - \cos(kL) \in \sigma(\Delta)$, then by Theorem \ref{thm: equivalence} so is $1 - \cos(2n\pi + kL)$, and hence $\left(\frac{2n\pi}{L} + k\right)^2 \in \sigma(\La)$ for all $n \in \N$. Similarly, $\left(\frac{2n\pi}{L} - k\right)^2$ is also in $\sigma(\La)$. For $k_1 = 0$, we know from Lemma \ref{lemma: Dirichlet} that the multiplicity of $\left(\frac{2n\pi}{L}\right)^2$ is $(\beta - 1) + 2$. Similarly, we also know from Lemma \ref{lemma: Dirichlet} that the multiplicity of $\left(\frac{(2n+1)\pi}{L}\right)^2$ is $(\beta -1) + 2\dim \ker (\Delta - 2\UI)$. Observing that $\dim \ker (\Delta-2\UI) = \left|\left\{k_j \in K: k_j = \frac{\pi}{L}\right\}\right|$, we can write the quantum spectral zeta function of $\La$ as
\begin{align}
\Z(s) &= (\beta - 1)\sum_{n=1}^\infty \left(\frac{n\pi}{L}\right)^{-2s} + 2\sum_{n=1}^\infty \left(\frac{2n\pi}{L}\right)^{-2s}\nonumber\\
&\hspace{1cm} + \sum_{j=2}^{V} \left(\sum_{n=0}^\infty \left(\frac{2n\pi}{L} + k_j\right)^{-2s} + \sum_{n=1}^\infty \left(\frac{2n\pi}{L} - k_j\right)^{-2s}\right)
\end{align}
since this is a rearrangement of a series that converges absolutely for $\Re(s) > \frac{1}{2}$.  Hence,
\begin{equation}\label{eq: qzf unsimplified}
\Z(s)= (4^s(\beta-1) + 2)\left(\frac{L}{2\pi}\right)^{2s}\zeta_R(2s) + \sum_{j=2}^{V} \left(\sum_{n=0}^\infty \left(\frac{2n\pi}{L} + k_j\right)^{-2s} + \sum_{n=1}^\infty \left(\frac{2n\pi}{L} - k_j\right)^{-2s}\right)
\end{equation}
where $\zeta_R(z)$ is the Riemann zeta function.

\begin{lemma}\label{lemma: Hurwitz}
For $\Re(s) < 0$ and $k_j \in \left.\left(0,\frac{\pi}{L}\right.\right]$, 
\begin{align}
&\sum_{j=2}^{V} \left(\sum_{n=0}^\infty \left(\frac{2n\pi}{L} + k_j\right)^{-2s} + \sum_{n=1}^\infty \left(\frac{2n\pi}{L} - k_j\right)^{-2s}\right)\nonumber\\
&\hspace{.5cm} = \frac{2L^{2s}\Gamma(1-2s)}{\pi}\sin(s\pi)\sum_{n=1}^\infty n^{2s-1}\sum_{j=2}^{V} \cos(k_jLn).
\end{align}
\end{lemma}

\begin{proof}
We can see by rearranging the left-hand side that
\begin{align}
 & \displaystyle\sum _{j=2}^{V}\left(\sum_{n=0}^\infty \left(\frac{2n\pi}{L} + k_j\right)^{-2s} + \sum_{n=1}^\infty \left(\frac{2n\pi}{L} - k_j\right)^{-2s}\right)\\
       &\hspace{.5cm}= \displaystyle \left(\frac{L}{2\pi}\right)^{2s}\sum _{j=2}^{V}\left(\sum_{n=0}^\infty \left(n + \frac{k_jL}{2\pi}\right)^{-2s} + \sum_{l=0}^\infty \left(l + \left(1 - \frac{k_jL}{2\pi}\right)\right)^{-2s}\right)\\
       &\hspace{.5cm}= \displaystyle \left(\frac{L}{2\pi}\right)^{2s}\sum_{j=2}^{V} \left(\zeta_H\left(2s,\frac{k_jL}{2\pi}\right) + \zeta_H\left(2s,1-\frac{k_jL}{2\pi}\right)\right)\label{eq: Hurwitz}
\end{align}
where $\zeta_H(z,a)$ is the Hurwitz zeta function. The Hurwitz zeta function is convergent for $\Re(z) > 1$ and $a > 0$. The second sum was rewritten to ensure that $a = 1-\frac{k_jL}{2\pi}$ is positive, and then \eqref{eq: Hurwitz} is analytic for $s \neq 1/2$ \cite{NITS}.  In the restricted domain $\Re(z) < 0$ and $0 < a \leq 1$,
\begin{equation}
\zeta_H(z,a) = \displaystyle\frac{2\Gamma(1-z)}{(2\pi)^{1-z}}\left[\sin\left(\frac{z\pi}{2}\right)\sum_{n=1}^\infty \frac{\cos(2\pi na)}{n^{1-z}} + \cos\left(\frac{z\pi}{2}\right)\sum_{n=1}^\infty \frac{\sin(2\pi na)}{n^{1-z}}\right].
\end{equation}
 Therefore
\begin{equation}
  \zeta_H\left(2s,\frac{k_jL}{2\pi}\right) = \displaystyle\frac{2\Gamma(1-2s)}{(2\pi)^{1-2s}}\left[\sin\left(s\pi\right)\sum_{n=1}^\infty \frac{\cos(k_jLn)}{n^{1-2s}} + \cos\left(s\pi\right)\sum_{n=1}^\infty \frac{\sin(k_jLn)}{n^{1-2s}}\right]
\end{equation}
and
\begin{equation}
   \zeta_H\left(2s,1-\frac{k_jL}{2\pi}\right) = \frac{2\Gamma(1-2s)}{(2\pi)^{1-2s}}\left[\sin\left(s\pi\right)\sum_{n=1}^\infty \frac{\cos(k_jLn)}{n^{1-2s}} - \cos\left(s\pi\right)\sum_{n=1}^\infty \frac{\sin(k_jLn)}{n^{1-2s}}\right].
\end{equation}
Combining these equations, we see that
\begin{equation}\label{eq: sum of Hurwitz}
 \zeta_H\left(2s,\frac{k_jL}{2\pi}\right)+\zeta_H\left(2s,1-\frac{k_jL}{2\pi}\right) = \frac{2\Gamma(1-2s)(2\pi)^{2s}}{\pi}\sin(s\pi)\sum_{n=1}^\infty n^{2s-1}\cos(k_jLn)
\end{equation}
for $\Re(s) < 0$. Substituting \eqref{eq: sum of Hurwitz} into \eqref{eq: Hurwitz} and rearranging the absolutely convergent sum completes the proof.
\end{proof}

Finally, we can write $\displaystyle\sum_{j=2}^{V} \cos(k_jLn)$ in terms of $Z(s)$, the spectral zeta function on the corresponding discrete graph.

\begin{lemma}\label{lemma: cheb}
Let $K=\{k_j \}_{j=1}^V$ be as defined previously, so $1-\cos (k_j L) \in \sigma(\Delta)$. Then
\begin{equation}
\sum_{j=2}^{V} \cos(k_jLn) = \sum_{r=0}^n (-2)^r\frac{n}{n+r}\binom{n+r}{2r}Z(-r)
\end{equation}
where $Z(s)$ is the discrete spectral zeta function.
\end{lemma}

\begin{proof}
One property of Chebyshev polynomials is the following \cite{NITS}:
\begin{equation}\label{eq: cosny}
\cos(y n) = T_n(\cos(y)).
\end{equation}
The $n^{th}$ Chebyshev polynomial can be written as
\begin{equation}\label{eq: Cheb}
T_n(x) = n\displaystyle\sum_{r=0}^n (-2)^r \frac{(n+r-1)!}{(n-r)!(2r)!} (1-x)^r = \sum_{r=0}^n (-2)^r \frac{n}{n+r}\binom{n+r}{2r}(1-x)^r.
\end{equation}
Combining equations (\ref{eq: cosny}) and (\ref{eq: Cheb}), we see that
\begin{equation}\label{eq: sumcosny}
  \cos(k_jLn) = \displaystyle\sum_{r=0}^n(-2)^r\frac{n}{n+r}\binom{n+r}{2r}(1 - \cos(k_jL))^r,
\end{equation}
and therefore, using \eqref{eq: discrete spectral zeta},
\begin{equation}
 \displaystyle \sum_{j=2}^{V} \cos(k_jLn) = \displaystyle \sum_{r=0}^n (-2)^r\frac{n}{n+r}\binom{n+r}{2r}Z(-r).
\end{equation}
\end{proof}

Merging Lemmas \ref{lemma: Hurwitz} and \ref{lemma: cheb}, we have shown that
\begin{align}\label{eq: z2}
&\sum_{j=2}^{V} \left(\sum_{n=0}^\infty \left(\frac{2n\pi}{L} + k_j\right)^{-2s} + \sum_{n=1}^\infty \left(\frac{2n\pi}{L} - k_j\right)^{-2s}\right)\nonumber\\
&\hspace{.5cm} = \displaystyle \frac{2L^{2s}\Gamma(1-2s)}{\pi}\sin(s\pi)\sum_{n=1}^\infty\sum_{r=0}^n (-2)^r\frac{n^{2s}}{n+r}\binom{n+r}{2r}Z(-r).
\end{align}
Notice that the sums cannot be interchanged because the convergence of the infinite sum is not absolute. Substituting \eqref{eq: z2} into \eqref{eq: qzf unsimplified} completes the proof of Theorem \ref{thm: main}.

\section{Applications}\label{sec: applications}

In this section, we demonstrate the usefulness of Theorem \ref{thm: main} by calculating the quantum spectral zeta function of an equilateral complete bipartite graph. We then use that spectral zeta function to compute the vacuum energy and spectral determinant associated with the Laplace operator on this quantum graph. Finally, we compare these results with previous results for a particular case: a star graph.

A \emph{complete bipartite graph}, denoted by $K_{m,p}$, is a graph with $m + p$ vertices and $mp$ edges whose vertices can be divided into two disjoint sets, set $\mathcal{U}$ of size $m$ and set $\mathcal{W}$ of size $p$, such that each vertex in set $\mathcal{U}$ is connected to every vertex in set $\mathcal{W}$; see Figure \ref{fig: complete bipartite}. The first Betti number of $K_{m,p}$ is $\beta = E - V + 1 = mp - (m + p) + 1$. The eigenvalues of a discrete complete bipartite graph $K_{m,p}$ are known to be 0, 1 (with multiplicity $m+p-2$), and 2 \cite{Chungbook}, and therefore
\begin{equation}
Z(s) = \sum_{j=1}^{m+p}\,\!^{^\prime} \lambda_j^{-s} = (m+p-2) + 2^{-s}.
\end{equation}

\begin{figure}[tbh]
\begin{center}
\begin{tikzpicture}
  [scale=1.5,every node/.style={circle,fill=black!, scale=.5}]
  \node (n1) at (0,.5) {};
  \node (n2) at (0,-.5)  {};
  \node (n3) at (1,1)  {};
  \node (n4) at (1,0) {};
  \node (n5) at (1,-1) {};

  \foreach \from/\to in {n1/n3,n1/n4,n1/n5,n2/n3,n2/n4,n2/n5}
    \draw (\from) -- (\to);

\end{tikzpicture}
\end{center}
\caption{\small{The complete bipartite graph $K_{2,3}$.}}\label{fig: complete bipartite}
\end{figure}
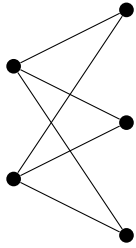 

By Theorem 1, the quantum spectral zeta function of the complete bipartite graph $K_{m,p}$, where each edge has length $L$, is given by
\begin{align}
\Z(s) &= \displaystyle \frac{2L^{2s}\Gamma(1-2s)}{\pi}\sin(s\pi)\sum_{n=1}^\infty\sum_{r=0}^n (-2)^r\frac{n^{2s}}{n+r}\binom{n+r}{2r}(m + p - 2 + 2^r)\nonumber\\
&\hspace{.5cm} + \left(4^s(mp - m - p) + 2\right) \left(\frac{L}{2\pi}\right)^{2s}\zeta_R(2s) \ .
\end{align}
Using the Chebyshev polynomial property \eqref{eq: Cheb}, the double summation can be written as
\begin{align}\label{eq: summation}
& \displaystyle (m+p-2)\sum_{n=1}^\infty\sum_{r=0}^n (-2)^r\frac{n^{2s}}{n+r}\binom{n+r}{2r}\left(1-\cos\left(\frac{\pi}{2}\right)\right)^{r}\nonumber\\
&\hspace{.5cm} + \sum_{n=1}^\infty\sum_{r=0}^n (-2)^r\frac{n^{2s}}{n+r}\binom{n+r}{2r}(1-\cos(\pi))^r\nonumber\\
&= \displaystyle (m+p-2)\sum_{n=1}^\infty n^{2s-1} T_n\left(\cos\left(\frac{\pi}{2}\right)\right) + \sum_{n=1}^\infty n^{2s-1}T_n(\cos(\pi))\nonumber\\
&= (m+p-2)\sum_{n=1}^\infty (-1)^n (2n)^{2s-1} + \sum_{n=1}^\infty (-1)^n n^{2s-1}
\end{align}
since we know by \eqref{eq: cosny} that $T_n(\cos(x)) = \cos(nx)$. Continuing, we can write \eqref{eq: summation} as
\begin{equation}
-(m+p-2)2^{2s-1}\eta(1-2s) - \eta(1-2s)
\end{equation}
where
\begin{equation}
\eta(z) = \sum_{n=1}^\infty \frac{(-1)^{n-1}}{n^z} 
\end{equation}
is the Dirichlet eta function. Note that
\begin{equation}\label{eq: Dirichlet eta function}
\eta(z) = (1-2^{1-z})\zeta_R(z) \hspace{1cm}\mbox{for } z \neq 1
\end{equation}
where $\zeta_R(z)$ is the Riemann zeta function \cite{NITS}. From this, we see that the quantum spectral zeta function of the equilateral complete bipartite graph $K_{m,p}$ is
\begin{align}\label{eq: qszf complete bipartite}
\Z(s) &= -\displaystyle \frac{L^{2s}\Gamma(1-2s)}{\pi}\sin(s\pi)[(m+p-2)4^s + 2]\eta(1-2s)\nonumber\\ 
&\hspace{.5cm}+ \left(4^s(mp-m-p) + 2\right) \left(\frac{L}{2\pi}\right)^{2s}\zeta_R(2s).
\end{align}

\subsection{Vacuum energy}

The vacuum energy of the Laplace operator is formally half the sum of the square roots of the eigenvalues,
\begin{equation*}
\frac{1}{2}\sum_{j=1}^\infty\,\!^{^\prime} k_j \ .
\end{equation*}
Hence, the zeta function regularization of the vacuum energy is
\begin{equation}
 E_c = \frac{1}{2}\Z(-1/2) \ .
\end{equation}
Given the form of $\Z(s)$ from equation \eqref{eq: qszf complete bipartite} (which is valid at $s = -\frac{1}{2}$ by Theorem \ref{thm: main}), 
\begin{align}\label{eq: vacuum energy}
E_c &=  \displaystyle \frac{\Gamma(2)}{2\pi L}\left[\frac{m+p-2}{2} + 2\right]\eta(2) + \left(\frac{mp-m-p}{2} + 2\right) \left(\frac{\pi}{L}\right)\zeta_R(-1)\nonumber\\
&= \frac{(3(m+p) - 2mp -6)\pi}{48L}
\end{align}
since $\Gamma(2) = 1$, $\zeta_R(-1) = -\frac{1}{12}$, and by \eqref{eq: Dirichlet eta function} $\eta(2) = \frac{1}{2}\zeta_R(2) = \frac{\pi^2}{12}$.

\subsection{Spectral determinant}

The spectral determinant of an operator is formally the product of its eigenvalues,
\begin{equation*}
\prod_{j=1}^\infty \hspace{-.05cm}\phantom{|}^{\prime} k_j^2 \ .
\end{equation*}
The zeta regularized spectral determinant is consequently defined as,
\begin{equation}
{\det}^\prime(\La) = \mbox{exp}(-\Z'(0)).
\end{equation}
 Using Theorem \ref{thm: main}, we know that for $\Re(s) < 0$, the quantum spectral zeta function of the complete bipartite graph $K_{m,p}$ is given by \eqref{eq: qszf complete bipartite}. However, the first term, which inherited the restriction to $\Re(s)<0$ from \eqref{eq: sum of Hurwitz}, is zero at $s=0$. Therefore, the derivative of the first term of \eqref{eq: qszf complete bipartite} at $s=0$ is 
\begin{align}\label{eq: D1}
-\Gamma(1)[(m+p-2) + 2]\eta(1) = -(m+p)\ln(2)
\end{align} 
since $\Gamma(1) = 1$ and $\eta(1) = \ln(2)$. The derivative of the second term at $s = 0$ is 
\begin{align}\label{eq: D2}
&\ln(4)(mp-m-p)\zeta_R(0) + 2(mp-m-p+2)\ln\left(\frac{L}{2\pi}\right)\zeta_R(0) + 2(mp-m-p+2) \zeta_R'(0)\nonumber\\
&\hspace{.5cm}= -\frac{\ln(4)(mp-m-p)}{2} - (mp-m-p+2)\ln(L)
\end{align}
since $\zeta_R(0) = -\frac{1}{2}$ and $\zeta_R'(0) = -\frac{\ln(2\pi)}{2}$. Hence, combining \eqref{eq: D1} and \eqref{eq: D2} we can see that the spectral determinant is
\begin{equation}\label{eq: spectral determinant}
{\det}^\prime(\La) = 2^{mp}L^{mp-m-p+2}.
\end{equation}

\subsection{Comparison with known results}\label{sec: comparison}

Here we consider the special case of a star graph. A \emph{star graph} with $E$ edges has $E$ vertices of degree one connected to a central vertex, and hence, $V=E+1$ vertices in all; see Figure \ref{fig: star}. In particular, a star graph with E edges is the complete bipartite graph $K_{1,E}$. 

Consider the equilateral quantum star graph where each edge has length $L$. We assume that the coordinate $x_e = 0$ at the vertices of degree one and $x_e = L$ at the center. The vertex conditions \eqref{eq: NK vertex conditions} become $f_e'(0) = 0$ and $\sum_{e\in\mathcal{E}} f_e'(L) = 0$.  Consequently, an eigenfunction has the form
$f_e(x_e) = a_e\cos \left(\frac{n\pi x_e}{L}\right)$ on edge $e$ and the Dirichlet eigenvalues $\{\left(\frac{n\pi}{L}\right)^2\}_{n=1}^\infty$ are in the spectrum of $\La$ with multiplicity one, which agrees with Lemma \ref{lemma: Dirichlet}.

\begin{figure}[tbh]
\begin{center}
\begin{tikzpicture}
  [scale=1.5,every node/.style={circle,fill=black!, scale=.5}]
  \node (n1) at (0,0) {};
  \node (n2) at (.3,.95)  {};
  \node (n3) at (-.8,.59)  {};
  \node (n4) at (-.8,-.59) {};
  \node (n5) at (.3,-.95) {};
  \node (n6) at (1,0)  {};

  \foreach \from/\to in {n1/n2,n1/n3,n1/n4,n1/n5,n1/n6}
    \draw (\from) -- (\to);

\end{tikzpicture}
\end{center}
\caption{\small{A star graph with $6$ vertices and $5$ edges; this is $K_{1,5}$.}}\label{fig: star}
\end{figure}
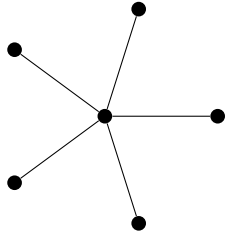

From \eqref{eq: vacuum energy}, we can see that for a star graph with $E$ edges (i.e., $K_{1,E}$),
\begin{equation}
E_c = \frac{\pi}{48L}(E - 3).
\end{equation}
 This agrees with the result of \cite{Full07} (also see \cite{BerHarWil09, HarKir11}) where the vacuum energy of an equilateral star graph was calculated. As the Casimir force is proportional to the derivative of $E_c$ with respect to $L$, this example was used to demonstrate that the Casimir force changes from attractive to repulsive depending on the number of edges.

From \eqref{eq: spectral determinant}, we can see that for a star graph with $E$ edges
\begin{equation}
{\det}^\prime(\La) = 2^EL.
\end{equation}
This agrees with the results of \cite{HarKir11} where the spectral zeta function and spectral determinant of a general quantum graph were computed using contour integrals.

\subsection*{Acknowledgements}
The authors would like to thank Gregory Berkolaiko and the anonymous referees, whose suggestions substantially simplified the presentation of the main result.    JH would like to thank the University of Warwick for their hospitality during his sabbatical where some of the work was carried out.  JH was supported by the Baylor University research leave program.  This work was partially supported by a grant from the Simons Foundation (354583 to Jonathan Harrison).

\bibliographystyle{abbrv}

\begin{thebibliography}{99}

\bibitem{Bipartitebook}
Asratian, A.~S., Denley, T.~M.~J. and H\"aggkvist, R.:
\newblock Bipartite graphs and their applications,
\newblock Cambridge University Press, Cambridge, U.K. (1998)

\bibitem{Below85}
von Below, J.:
\newblock A characteristic equation associated to an eigenvalue problem on {$c^2$}-networks,
\newblock Linear Algebra Appl. 71:309-325 (1985)

\bibitem{BerHarWil09}
Berkolaiko, G.,  Harrison, J.~M. and  Wilson, J.~H.:
\newblock Mathematical aspects of vacuum energy on quantum graphs,
\newblock  J. Phys. A 42(2):025204, 20 (2009)

\bibitem{BKbook}
Berkolaiko, G. and Kuchment, P.:
\newblock  Introduction to quantum graphs, vol. 186 of {\em Mathematical
  Surveys and Monographs}.
\newblock American Mathematical Society, Providence, RI (2013)

\bibitem{Chungbook}
Chung, F.:
\newblock Spectral graph theory, vol. 92 of {\em CBMS Regional Conference Series in Mathematics}. 
\newblock American Mathematical Society, Providence, RI (1997)

\bibitem{FriKar16}
Friedli, F. and Karlsson, A.:
\newblock Spectral zeta functions of graphs and the {R}iemann zeta function in the critical strip,
\newblock  preprint arXiv:1410.8010 (2016)

\bibitem{Full07}
Fulling, S.~A.,  Kaplan, L. and  Wilson, J.~H.:
\newblock Vacuum energy and repulsive {C}asimir forces in quantum star graphs,
\newblock  Phys. Rev. A (3) 76(1):012118, 7 (2007)

\bibitem{FulKucWil07}
Fulling, S.~A., Kuchment, P. and Wilson, J.~H.:
\newblock Index theorems for quantum graphs,
\newblock J. Phys. A 40(47):14165, 16  (2007)

\bibitem{HarKir11}
 Harrison, J.~M. and Kirsten, K.:
\newblock Zeta functions of quantum graphs,
\newblock  J. Phys. A 44(23):235301, 29 (2011)

\bibitem{HarKirTex12}
Harrison, J.~M., Kirsten, K. and Texier, C.:
\newblock Spectral determinants and zeta functions of Schr{\"o}dinger operators on metric graphs,
\newblock J. Phys. A 45(12):125206, 14 (2012)

\bibitem{HarWeyKir16}
Harrison, J.~M., Weyand, T. and Kirsten, K.:
\newblock Zeta functions of the {D}irac operator on quantum graphs,
\newblock J. Math. Phys. 57:102301, 10 (2016)

\bibitem{Has89}
Hashimoto, K.:
\newblock Zeta functions of finite graphs and representation of p-adic groups,
\newblock Adv. Studies Pure Math 15:211-280 (1989)

\bibitem{Kuc04}
Kuchment, P.:
\newblock Quantum graphs {I}. {S}ome basic structures,
\newblock Waves Random Media 14(1) S107--S128 (2004)

\bibitem{NITS}
 Olver, F.~W.~J., Lozier, D.~W., Boisvert, R.~F. and  Clark, C.~W. (eds):
\newblock  N{IST} handbook of mathematical functions,
\newblock U.S. Department of Commerce, National Institute of Standards and
  Technology, Washington, DC. Cambridge University Press, Cambridge, U.K. (2010)

\bibitem{Pank06}
Pankrashkin, K.:
\newblock Spectra of {S}chr\"odinger operators on equilateral quantum graphs,
\newblock  Lett. Math. Phys. 77(2) 139--154 (2006)

\bibitem{StaTer96}
Stark, H.~M. and Terras, A.~A.:
\newblock Zeta functions of finite graphs and coverings,
\newblock Adv. Math. 121:142--165 (1996)

\bibitem{Sun86}
Sunada, T.:
\newblock {$L$}-functions in geometry and some applications.
\newblock In:  Shiohama, K., Sakai, T. and Sunada, T. (eds), Curvature and topology of {R}iemannian manifolds ({K}atata,
  1985), Lecture Notes in Math. vol. 1201, pp. 266--284.
  Springer, Berlin (1986)


\end{thebibliography}

\Address
\end{document}